\documentclass[11pt, a4paper]{amsart}

\usepackage{amsmath}
\usepackage{latexsym}
\usepackage[all]{xy}
\usepackage{color}
\newtheorem{teorema}{Theorem}[section]
\newtheorem{definicion}[teorema]{Definition}
\include{amslatex}
\newtheorem{proposicion}[teorema]{Proposition}
\newtheorem{lema}[teorema]{Lemma}
\newtheorem{corolario}[teorema]{Corollary}

\newtheorem{comentario}[teorema]{Remark}

\numberwithin{equation}{section}

\begin{document}
\begin{title}[Maximal acceleration geometries and curvature bounds]
 {Maximal acceleration geometries and spacetime curvature bounds}
\end{title}
\clearpage\maketitle
\thispagestyle{empty}
\begin{center}
\author{Ricardo Gallego Torrom\'e\footnote{Email: rigato39@gmail.com}}
\end{center}

\begin{center}
\address{Department of Mathematics\\
Faculty of Mathematics, Natural Sciences and Information Technologies\\
University of Primorska, Koper, Slovenia}
\end{center}
\begin{abstract}
A geometric framework for metrics of maximal acceleration which is applicable to large proper accelerations is discussed, including a theory of connections associated with the geometry of maximal acceleration. In such a framework it is shown that the uniform bound on the proper maximal acceleration implies an uniform bound for certain bilinear combinations of the Riemannian curvature components in the domain of the spacetime where curvature is finite.
\end{abstract}

\section{Introduction}
The conjecture on the existence of an universal or uniform bounds on proper acceleration has attracted the attention of researches for a long time \cite{Brandt1983,Caianiello,Ricardo Nicolini 2018}. The hypothesis of maximal proper acceleration was first discussed by E. Caianiello \cite{Caianiello}
in the context of a geometric approach to the foundations of the quantum theory \cite{Caianielloquantum}. As a consistence requirement for the positiveness in the mass spectra of quantum particles and the existence of a maximal speed, Caianiello found a positiveness condition for a Sasaki-type metric in the phase space description of quantum mechanics. Such condition leaded to the existence of a maximal proper acceleration depending on the mass of the particle. In classical models of gravity, the consequences of the existence of a maximal proper acceleration have been studied extensively. Let us mention
for instance the investigation of maximal proper acceleration for Rindler spaces \cite{CaianielloFeoliGasperiniScarpetta},
Schwarzschild \cite{FeoliLambiasePapiniScarpetta}, Reissner-Nordst$\ddot{\textrm{o}}$m \cite{BozzaFeoliPapiniScarpetta},
Kerr-Newman \cite{BozzaFeoliLambiasePapiniScarpetta} and Friedman-Lema\^{i}tre metrics
\cite{CaianielloGasperiniScarpetta}, among other investigations.  Independently, in the theory developed by H. Brandt, the starting point is the energy-time uncertainty relation, that combined with an argument involving the idea of a breakdown of the topological and smooth structures of spacetime at the Planck scale, implies the existence of an universal maximal proper acceleration  \cite{Brandt1983}. This approach developed further into a effective tangent spacetime geometry  \cite{Brandt1989}. These two theories have interesting consequences and initiated the study of other theories where maximal proper acceleration is uniformly upper bounded \cite{Ricardo Nicolini 2018}.

Caianiello's and Brandt's theories could be classified as quantum mechanical motivated frameworks. In a different research line, maximal proper acceleration has been related with the foundations of the theory of special and general relativity as follows \cite{Mashhoon1990}. It is well known that when developing the geometric framework for relativistic theories, an hypothesis on the characteristics of ideal clocks, namely, the clock hypothesis, is usually adopted. It states  that ideal clocks do not depend upon the acceleration suffered by the clock \cite{Einstein1922, Syngespecial1965}. The clock hypothesis, in conjunction with the rest of principles of relativistic theories, lead to the conclusion that spacetime geometry must be of Finslerian type (\cite{Syngespecial1965}, chapter I), among which Lorentzian geometry is a prominent example. However, the adoption of the clock hypothesis is logically unjustified for generic dynamical situations, even if it leads to a simplification in the theoretical treatment of the properties and behavior of ideal clocks in accelerated motion \cite{Einstein1922}. Furthermore, it has been convincingly argued that in situations where radiation reaction is not negligible, the clock hypothesis is not applicable \cite{Mashhoon1990}, leading to the idea that spacetime geometry must be described by a geometric structure depending also on the acceleration of the probe particles world line curves.

Motivated by the problem of finding a rigourous mathematical formalism for Caianiello's theory of metrics with maximal acceleration \cite{Ricardo2007} and the problem of radiation reaction in classical electrodynamics, an approach to spacetime metric structures compatible with maximal proper acceleration was proposed in a geometric framework of higher order jet geometry \cite{Ricardo2012,Ricardo2015}. This axiomatic approach assumes directly that the clock hypothesis does not hold in certain relevant dynamical domains and as a consequence of further several natural assumptions, the spacetime metrics must be a higher order jet spacetime metric. It is also assumed that when the effects due to acceleration are negligible, the spacetime metric must respect the clock hypothesis with high accuracy, which leads in the simplest case, to Finslerian/Lorentzian models of spacetime. For higher order jet geometry models the spacetime metric structure depends on {\it how spacetime is probed by test point particles}. This implies that the spacetime metric models depend at least upon the second jet bundle of the world lines of test particles. We will restrict our considerations to models of spacetime whose metric components depend on the second derivatives only, since this is enough to provide a consistent classical models of radiation-reaction systems \cite{Ricardo2017}.

In contrast with others approaches to the foundations of maximal acceleration geometry \cite{CaianielloFeoliGasperiniScarpetta,Brandt1989}, our approach has the advantage of not imposing a duplicity of metric structures (the metric of maximal acceleration and the usual Lorentzian spacetime), either explicitly or implicitly through a dynamical mechanism generating the maximal acceleration. The only fundamental structure is the higher order jet metric, from where the Lorentzian metric $g_0$ is obtained as a formal limit. When the model is made explicit as a power series on the inverse of the maximal acceleration \cite{Ricardo2015}, fundamental physical principles in addition to further formal assumptions implies that the metric has a particular form, that we have called {\it metric of maximal acceleration}. This theory is discussed in section 2.

 The intuitive implication of a maximal proper acceleration for spacetime curvature is the existence of an uniform upper bound on the Riemannian curvature tensor of the Lorentzian metric limit and that such a bound is valid in the region of the spacetime where the curvature is finite. Precedents of this result are found, for instance, in the context of Caianiello's quantum geometry, where it has been argued how  maximal acceleration implies the regularization of the big bang singularity  \cite{CaianielloGasperiniScarpetta, Gasperini 1987, Gasperini 1991}. Recently, it has been argued in the context of loop quantum gravity that maximal acceleration implies the resolution of singularities \cite{RovelliVidotto}. In this paper we confirm these insights, showing how in  spacetimes with metrics of maximal proper acceleration as discussed in \cite{Ricardo2015} and developed with higher clarity and deepness in the present paper, there are uniform bounds in certain combinations of the components of Riemannian curvature endomorphims of the associated Lorentzian limit metric. We achieve related results in a rigourous way by extending the uniform bound of maximal proper accelerations from the perturbative regime discussed in \cite{Ricardo2015} to large accelerations and by extending the uniform bound on the proper accelerations to certain relative accelerations associated to Jacobi fields.

The structure of this paper is the following. In {\it section 2}, the notion of {\it spacetime with a metric of maximal acceleration} is introduced in a general framework of higher order jet geometries. The treatment presented here is based upon fundamental principles and supersedes the perturbative approach discussed in previous works \cite{Ricardo2007, Ricardo2015}. We give a precise definition of the notion of maximal proper acceleration. In {\it section 3} it is recalled how curvature and relative acceleration along a geodesic are related through the notion of Jacobi field as solution of the Jacobi geodesic deviation equation.  The proper acceleration is an acceleration along a causal curve, while the notion of relative acceleration is a vector field along a given geodesic. However, it is shown how maximal acceleration can also be applied to this notion of relative acceleration. Such extension is achieved by means of considering certain aggregate of timelike world lines whose acceleration coincide pointwise with the relative acceleration associated to certain solutions of the Jacobi equation. Hence the uniform bound on maximal proper acceleration as discussed in {\it section 2} can be applied and as a consequence, an uniform bound on certain combinations of the curvature components spacetime curvature must hold in the region of the spacetime where curvature is finite.

The uniform bound on the curvature applies to the Riemannian curvature of the Lorentzian limit associated to the metric of maximal acceleration. The significance for the metric of maximal acceleration is clarified in {\it section 4} and {\it section 5}, where the foundations for a geometric theory of spacetimes with maximal acceleration are investigated. Starting with a discussion of a new version of the equivalence principle in spacetimes of metrics with maximal acceleration, we develop a theory of connections compatible with this new form of the equivalence principle and with the notion of spacetime with maximal acceleration. Then we describe an axiomatic characterization of the connection in terms of structure equations.

In the discussion {\it section}, several general remarks on the theory developed in this paper are highlighted. We also discuss the relation of our theory with Caianiello's theory of maximal acceleration \cite{Caianiello, CaianielloFeoliGasperiniScarpetta} and with Brandt's theory \cite{Brandt1983,Brandt1989}. Although formally the metric of maximal acceleration is the same for the three theories, our approach diverges considerably from the respective approaches of these authors. However, the bound on the curvatures also hold for the corresponding metrics, if the assumption of maximal acceleration as equivalent to a minimal proper time \cite{Caldirola1981} is accepted. Several future research lines are briefly indicated.
\section{Metrics of maximal acceleration}
Let us first introduce the geometric framework for this work. The spacetime manifold $M_4$ will be a $four$-dimensional smooth manifold.
 We assume that spacetime is classical, that is, the spacetime models are not allowed to be in superposition of spacetime structures. Therefore, in our theory, the spacetime structure does not suffer of quantum fluctuations. For such spacetime, the theoretically simplest, but also universal, procedure for testing the structure of the spacetime is by observing the smooth world lines $\{x:I\to M_4\}$ associated to point test particles\footnote{There is also the theoretical possibility to use test fields. Then it comes the problem of how measure test fields. We are adopting here the reductionist point of view that phenomenologically, fields can be characterized by their effects on idealized test particles effect on idealized test point particles.}. The metric structures that we will consider were named {\it metrics of maximal proper acceleration} or {\it metrics of maximal acceleration} \cite{Ricardo2015}.
 A metric of maximal proper acceleration $g$ on $M_4$ is a spacetime structure  that depends  upon the second jet of smooth world lines. It determines the proper time that an ideal observer will measure along its world line $x:I \to M_4$.

 In order to understand the notion of metric of maximal acceleration, it is preferable to develop first the fundamental concepts for spacetimes endowed with a higher order jet metric structure.
 If the curve $x:I\to M_4$ is (piecewise)-smooth causal, which means $g(x',x')\leq 0$, then the proper time that an ideal observer with world line $x:I\to M_4$ will measure is given by the expression
\begin{align}
\tau[x]=\int^t_{t_0} \,\Big[\,-g(x',x')\Big]^{\frac{1}{2}}\,ds,\quad t_0,t_1\,\in I.
\label{propertimeg}
\end{align}
For the metric structures that we will consider in this paper, the integrand $[\,-g(x',x')]^{\frac{1}{2}}$ is not homogeneous under changes of the integration parameter $s$. The higher order jet metric requires a precise specification of the parameter of integration in the definition of the proper time functional \eqref{propertimeg}. However, a weaker form of re-parametrization invariance is still available for the proposed models of metrics discussed below.

 For the spacetimes that we will consider in this paper, {\it Einstein clock hypothesis} \cite{Einstein1922} does not hold, since the proper time functional $\tau[x]$ depends on acceleration of the test particle world line  $x:I\to M_4$. A fundamental motivation to adopt spacetimes endowed with a metric  with maximal acceleration has been discussed by the author in \cite{Ricardo2012, Ricardo2015, Ricardo2017b} and concerns the problem of radiation reaction in classical electrodynamics, were it was argued that clock hypothesis must be violated \cite{Mashhoon1990}. Adopting geometries of maximal acceleration allows to resolve the problem of run-away solutions in electrodynamics. Indeed, dynamical laws consistent with a geometry of maximal acceleration and with the notion of higher order jet fields \cite{Ricardo2012,Ricardo2017b} are free of the problems of run-away solutions and pre-accelerated solutions, because the equation of motion consistent with such spacetimes is of second order and does not have run-away solutions \cite{Ricardo2017b}.

In the framework of higher order jet order fields, it is natural to adopt the following form for the integrand $g(x',x')$ in the expression for the proper time \eqref{propertimeg},
\begin{align}
 g_{^2x}(x',x')=\,g_0(x',x')\,+\xi_{(\,^2x,A^2_{\mathrm{max}})}(x',x'),
  \label{perturabativeexpansion}
 \end{align}
where the time parameter used when calculating the derivatives is the proper time associated to the metric structure $g_0$.
$^2x (t)$   is the second jet at $x(t)$ of the curve $x:I\to M_4$. In local coordinates, $^2 x(t)$ is represented by
 \begin{align*}
\,^2x^\mu (t)=\,(x^\mu(t),x^\mu \,' (t),x^\mu\, ''(t)),\quad \mu=0,1,2,3.
\end{align*}

The expressions \eqref{propertimeg}-\eqref{perturabativeexpansion} determine the general form of a proper time functional $\tau[x]$ compatible with our assumptions, namely,
 \begin{enumerate}
\item Dependence of the spacetime metric structure $g$ on the second jet bundle coordinates $\,^2x^\mu (t)=\,(x^\mu(t),x^\mu \,' (t),x^\mu\, ''(t))$ in a way that Einstein clock hypothesis does not hold in general,

\item The compatibility with Einstein clock hypothesis is recovered in the limit when the proper acceleration effects are negligible.
In the formalism described above, this is achieved if
\begin{align*}
\lim_{A^2_{\mathrm{max}}\to +\infty}\,g(\,^2x):=\,g_0,
\end{align*}
is compatible with the clock hypothesis.
\end{enumerate}
Compatibility with the clock hypothesis in the limit $A_{\mathrm{max}}\to +\infty$ implies that
$g_0$ lives on the first jet bundle $J^1 M_4\cong \,TM_4$. That the general form of a classical spacetime structure compatible with the fundamental principles of relativity needs to be of {\it Finsler type}\footnote{It was proved by the author that the most general geometric clock compatible with the clock hypothesis can be defined in the category of generalized Finsler spaces, or Lagrange spaces \cite{Ricardo2017}.} was discussed already by Synge \cite{Syngespecial1965}, Ch. I and also recently in \cite{Ricardo2017}.
However, in order to simplify our  treatment, we will consider the case when $g_0$ is a Lorentzian metric $\eta$.
In the expression \eqref{perturabativeexpansion}, $g_0(x',x')$ determines a spacetime metric compatible with Einstein clock hypothesis, while $\xi_{(\,^2x,A^2_{\mathrm{max}})}$ is the part of the metric $g$ that explicitly violates the clock hypothesis.

The existence of this metric $g_0$ is justified in the framework of spacetime with metrics of maximal acceleration, since one can probe the spacetime structure using geodesics, in which case, the covariant condition $D_{x'}\,x'=0$ holds, in accordance with the general philosophy of spacetimes with higher order jet metrics. In this case, the spacetime structure found is the Lorentzian metric $g_0$, where $D$ is the covariant derivative of an affine connection. Since it is already at our disposition the Lorentzian $g_0$ structure, $D$ is interpreted as the Levi-Civita connection of $g_0$.

The explicit violation of the Einstein clock hypothesis is parameterized by a function defined along the lift $^2x:I\to J^2M_4$
\begin{align*}
\xi_{(\,^2x,A^2_{\mathrm{max}})}(x',x'):\,^2x\to \mathbb{R},
 \end{align*}
 that also introduces the parameter $A_{\mathrm{max}}$. Since $\xi_{(\,^2x,A^2_{\mathrm{max}})}(x',x')$ must be negligible when the effects of acceleration are small, as we know from the precise validity of classical relativistic dynamics, a natural way to quantify the effects of acceleration on the spacetime structure is through the comparison of the relevant notion of acceleration  with respect to the deformation parameter $A^2_{\mathrm{max}}$. Hence the formal dependence of $\xi_{(\,^2x,A^2_{\mathrm{max}})}(x',x')$ on the acceleration must be through the quotient $g_0(D_{x'}x',D_{x'}x')/A^2_{\mathrm{max}}$.

  According to these considerations, the general expression for the higher order spacetime structure is of the form
\begin{align*}
g_{^2x}(x',x')=\,g_0(x',x')\,+\xi\left(\frac{g_0\left(D_{x'}x',\,D_{x'}x'\right)}{A^2_{\mathrm{max}}}\right).
\end{align*}
It is also natural to assume that the components of $g$ are analytical in $1/A^2_{\mathrm{max}}$,
\begin{align*}
\xi_{(\,^2x,A^2_{\mathrm{max}})}(x',x')=\,\sum^{+\infty}_{n=1}\,\xi_n(x)\,\left( \frac{g_0\left(D_{x'}x',\,D_{x'}x'\right)}{A^2_{\mathrm{max}}}\right)^n,
\end{align*}
with $\xi_n:M_4\to \mathbb{R}$ being spacetime functions.

Let us stress that when radiation reaction effects are of relevance, Einstein clock hypothesis is violated. We took as granted the reciprocal statement, namely, that if Einstein clock hypothesis holds good, then radiation reaction effects are negligible and as a consequence, the metric structure of the spacetime must be consistent with the hypothesis. Therefore, the function $\xi_{(\,^2x,A^2_{\mathrm{max}})}(x',x')$ must be monotonic on the argument $\frac{g_0\left(D_{x'}x',\,D_{x'}x'\right)}{A^2_{\mathrm{max}}}$. Otherwise, there will be local domains of acceleration, apart from the domain where the proper accelerations are small, where the metric $g$ effectively does not depend upon the ratio $\frac{g_0\left(D_{x'}x',\,D_{x'}x'\right)}{A^2_{\mathrm{max}}}$ and hence, the clock hypothesis  will still hold good. This kind of phenomena, however, will imply special scales of acceleration, apart from the small accelerations, where radiation reaction effects are negligible, a fact that seems difficult to physically justify. We are led to assume that in $\xi_{(\,^2x,A^2_{\mathrm{max}})}(x',x')$ each of the terms
\begin{align*}
\xi_n(x)\,\left( \frac{g_0\left(D_{x'}x',\,D_{x'}x'\right)}{A^2_{\mathrm{max}}}\right)^n,\quad n=1,2,3,...
\end{align*}
must be either zero or all the remain non-trivial terms must have the same sign.

Furthermore, if the functions $\xi_n:M_4 \to \mathbb{R}$ are not constant, then the notion of small acceleration needs to be compared not only against $A_{\mathrm{max}}$, but also against each of the functions $\xi_n$. Such comparisons imply the introduction of additional scales of accelerations. Also, if the functions $\xi_n:M_4 \to \mathbb{R}$ are not constant, each of them will need to have a very definite dynamics, with an addition of a large conceptual and technical complication in our models. Such difficulties are overcome if we assume that $\xi_n$ are all constants of the same sign.

According with the above argument, the expression for the function $\xi$ must have the following structure,
\begin{align}
\xi_{(\,^2x,A^2_{\mathrm{max}})}(x',x')=\,\sum^{+\infty}_{n=1}\,\xi_n\,\left( \frac{g_0\left(D_{x'}x',\,D_{x'}x'\right)}{A^2_{\mathrm{max}}}\right)^n,
\label{generaltypeofcorrections}
\end{align}
where each of the functions $\xi_n :M_4\to \mathbb{R}$ is constant and such that $\xi_n\,\xi_m \geq 0$ for all $m,n\in \,\mathbb{R}$. Depending on the sign of each of the functions $\xi_n$, we have two different families of functions:
\begin{enumerate}
\item
 $\xi_n >0$ for each $n\in\,\mathbb{N}$ or

\item $\xi_n <0$ for each $n\in\,\mathbb{N}$.
\end{enumerate}
These two classes of functions imply different models for $g$, whose relation with $g_0$ is either to decrease or to increase the absolute value from $|g_0(x',x')|$ to $|g(x',x')|$, respectively. We will see that the restriction in the preservation of the causal properties for $g$ and $g_0$ implies to consider the first family of metrics. Hence from now on, we will consider that the constants $\xi_n$ are all positive.

\subsection{Causal structure of the higher order jet metrics}
The {\it null bundle} associated to the Lorentzian metric $g_0$ is
\begin{align*}
Null_0:=\,\bigsqcup_{x\in M_4} Null_0(x),\quad Null_0(x):=\{y\in T_x M_4\,s.t.\,g_0(y,y)=0\}
 \end{align*}
 and the null bundle associate to the metric of maximal acceleration $g$ is
\begin{align*}
Null:=\,\bigsqcup_{x\in M_4} Null(x),\quad Null(x):=\{y\in T_x M_4\,s.t.\,g(y,y)=0\}.
 \end{align*}
 The projection $\pi_{00}:Null_0\to M_4$ is such that if $\zeta\in Null_0$, then $\zeta \in \,T_x M_4$ for a given $x\in M_4$ and $g_0(\zeta,\zeta)=0$, where $\pi_{00}(\zeta)=x$. The projection $\pi_{20}:Null\to M_4$ is defined similarly.
 Consider a section $Z\in\,\Gamma Null_0$. Then $\lambda Z\in\,\Gamma Null_0$ too, for any smooth function $\lambda:M_4\to \mathbb{R}$. The analogous property for sections of $Null$ is not true: if  $g(Z,Z)=0$, then in general $g(\lambda Z,\lambda Z)\neq 0$. Hence we have that
 \begin{align*}
 Null_0 \neq \,Null,
 \end{align*}
 since their sections do not coincide.
  However, if we only consider light-like curves which are geodesics of the Levi-Civita connection $D$, then we have the following result,
 \begin{proposicion}
 Let $x:I\to M_4$ be a geodesic of the Levi-Civita connection $D$. Then $g_0(x',x')=0$ iff $g(x',x')=0$.
 \label{comparison light structures}
 \end{proposicion}
 Therefore, when restricted to geodesic motion, the null structure of $g_0$ coincides with the null structure of $g$. If light rays follow geodesics of $D$, then in spacetimes with a higher order jet metric of the form \eqref{perturabativeexpansion} there is an unique light cone structure given by the null structure of the metric $g_0$.

 The following converse of Proposition \ref{comparison light structures} also holds good:
 \begin{proposicion}
In a spacetime $(M_4,g)$ such that $\xi$ is given by \eqref{generaltypeofcorrections}, the conditions  $g_0(x',x')=0$ and  $g(x',x')=0$ implies the geodesic condition $D_{x'}x'=0$.
 \end{proposicion}

 A {\it timelike curve} of the higher order jet metric \eqref{propertimeg} is a curve $x:I\to M_4$ such that
\begin{align*}
g_0(x',x')+ \sum^{+\infty}_{n=1}\,\xi_n\,\left( \frac{g_0\left(D_{x'}x',\,D_{x'}x'\right)}{A^2_{\mathrm{max}}}\right)^n<\,0.
\end{align*}
Analogously, a {\it spacelike curve} of the higher order jet metric \eqref{propertimeg} is a curve $x:I\to M_4$ such that
\begin{align*}
g_0(x',x')+ \sum^{+\infty}_{n=1}\,\xi_n\,\left( \frac{g_0\left(D_{x'}x',\,D_{x'}x'\right)}{A^2_{\mathrm{max}}}\right)^n >\,0.
\end{align*}
Let us   choose the parametrization of the curves in the initial definition of the functional $\tau[x]$ given by the expression \eqref{propertimeg} to be the proper time of the metric $g_0$. Since $D$ is the covariant derivative of the Levi-Civita connection of $g_0$, then the condition $g_0(x',x')=-1$ implies $g_0(D_{x'}x', x')=0$. Hence we have by a standard argument that $g_0(D_{x'}x',D_{x'} x')>0$.
Now, let us  note that since we have assumed that all the coefficients $\xi_n$ have the same sign, if $\xi_n\geq 0$ for each $n\in\,\mathbb{N}$, then we have the relation
\begin{align*}
g_{^2x}(x',x')& =\,-1\, +\sum^{+\infty}_{n=1}\,\xi_n\,\left( \frac{g_0\left(D_{x'}x',\,D_{x'}x'\right)}{A^2_{\mathrm{max}}}\right)^n >\,-1=\,g_0 (x',x').
\end{align*}
This relation implies the following result
\begin{proposicion}
 If the curve $x:I\to M_4$ is timelike respect to $g$, then it is timelike respect to $g_0$.
 \label{restriction on the causal structute of g}
  \end{proposicion}
  Therefore, the set of timelike curves of $g$ is a subset of the timelike curves respect to $g_0$. This implies that causality conditions respect to $g$ is stronger than respect to $g$. If we assume that causal worldlines  are either timelike or spacelike respect to $g$, then $\xi_n\geq 0$ assures that in the case of timelike world lines are included in the interior of relativistic light cones. We think that this property justify the choice of the sign for the non-zero $\xi_n$.

\begin{comentario}It is remarkable that the metric $g$ does not hold the so called {\it orthonormal condition}  $g(D_{x'}x',x')\,\neq 0$, since $g$ does not necessarily preserves the Levi-Civita connection $D$. The orthonormality condition is of fundamental relevance in the derivation of the Lorentz-Dirac equation in classical electrodynamics. However, the  failure of this condition in the case of higher order jet geometry allows to formulate of a second order differential equation for the electron \cite{Ricardo2017}.
 \end{comentario}

\subsection{The metric of maximal proper acceleration as a special case of higher order jet metric}
Let  $x:I\to M_4$ be a causal curve in the sense that $g(x',x')\leq 0$ and assume a parametrization such that $g_0(x',x')=-1$. Then it holds that
\begin{align}
\sum^{+\infty}_{n=1}\,\xi_n\,\left( \frac{g_0\left(D_{x'}x',\,D_{x'}x'\right)}{A^2_{\mathrm{max}}}\right)^n\leq\,1
\label{constrain on the sum}
\end{align}
in the range of validity of the analytical expression \eqref{perturabativeexpansion}.
According to the argument given above, let us assume that the constant $\xi_1$ is positive. Then $\xi_1$ can be re-absorbed within the metric\footnote{This operation also changes the value of the parameter $A^2_\mathrm{max}$ when it is associated to the maximal acceleration, but we will continue to use the same notation for the parameter $A_\mathrm{max}$ in the following.} $g_0$. Therefore, we can take without loss of generality the condition $\xi_1 =\,1$. If we denote by
 \begin{align*}
 \epsilon =\,\frac{g_0(D_{x'}x',D_{x'}x')}{A^2_{\mathrm{max}}},
 \end{align*}
 then the constrain \eqref{constrain on the sum} on $\xi_{(\,^2x,A_\mathrm{max})}(x',x')$  is of the form
 \begin{align}
 \epsilon+\,\sum^{+\infty}_{n=2}\,\xi_n \,\epsilon^n \leq\,1.
 \label{constrain on the sum 2}
 \end{align}
 It is reasonable that for  dynamical systems where interactions are strong, the value of $\epsilon$ can reach arbitrary close values to $1$. Then we assume that the validity of \eqref{perturabativeexpansion} is on the physical range $\epsilon\in\,[0,1[$, or at least, in a domain where $\epsilon\to 1^-$. Since all the terms in the left side of the expression \eqref{constrain on the sum 2} are positive, then the condition $\epsilon \to 1^-$ implies the limit conditions
 \begin{align*}
 \sum^{+\infty}_{n=2}\,\xi_n \,\epsilon^n \longrightarrow^{\epsilon\to 1^-} \sum^{+\infty}_{n=2}\,\xi_n \,1^- \to 0^+.
 \end{align*}
Since $\xi_n\geq 0$ for $n=2,3,...,$, this condition can only holds if
\begin{align}
\xi_n=0\, \,\forall, \, n\geq 2.
\end{align}

In this way, we arrive to a compact expression for the higher order jet metric of the type described by the expressions \eqref{perturabativeexpansion}, namely,
\begin{align}
g_{^2x} (x',x')=\,-\Big(1- \frac{  g_0(D_{x'}x',
D_{x'}x')}{A^2_{\mathrm{max}}}\Big),
\label{maximalaccelerationmetric0}
\end{align}
where a parametrization of the curve $x:I\to M_4$ such that $g_0(x',x')=-1$ has been used.
Generalizing the expression \eqref{maximalaccelerationmetric0}, {\it the metric of maximal acceleration} acting on two arbitrary vector fields $W,Q$ along $x:I\to M_4$ is a tensor field of order $(0,2)$ living on the second jet $^2x:I\to J^2M_4$ given by the expression
\begin{align}
g_{^2x}(W,Q)=\,\Big(1- \frac{  g_0(D_{x'}x',
D_{x'}x')}{A^2 _{\mathrm{max}}\,}\Big)g_0(W, Q).
\label{maximalaccelerationmetric}
\end{align}

The form \eqref{maximalaccelerationmetric0} is formally a {\it covariant version} \cite{Ricardo2007} of the metric thoroughly investigated by E. Caianiello and co-workers \cite{CaianielloFeoliGasperiniScarpetta}, but it has a natural interpretation as a {\it higher order jet geometry} defined by the relation \eqref{maximalaccelerationmetric} \cite{Ricardo2007, Ricardo2012, Ricardo2015}. From this second point of view, the metric structure of the spacetime assigns to each test particle probing the structure of the classical spacetime a line element which depends upon the second order time coordinates derivatives of the test particle's world line.

The proper time functional in a spacetime with a metric of maximal acceleration is of the form
\begin{align}
\tau[x]=\int^t_{t_0} \,\Big[1- \frac{  g_0(D_{x'}x'(s),
D_{x'}x'(s))}{A^2 _{\mathrm{max}}}\Big]^{\frac{1}{2}}\,ds,
\label{propertimeg0}
\end{align}
where $s$ is the proper time functional parameter calculated with $g_0$ along $x:I\to M_4$ and is defined by the functional
\begin{align*}
s[x]=\,\int^{\lambda_f}_{\lambda_0}\,[-g_0(\dot{x}(\lambda),\dot{x}(\lambda))]^{\frac{1}{2}}\,d\lambda .
\end{align*}
The functional $\tau[x]$ is not re-parametrization invariant in the usual sense, since the proper time parameter $t$ has been fixed to be the proper time of the metric $g_0$ and the integrand in $\tau[x]$ is not homogeneous. However, the functional $s[x]$ is invariant under re-parameterizations $\varphi:I_1\to I_2,\,\lambda \mapsto \tilde{\lambda}$.  This implies that in physical terms, also the functional $\tau[x]$ is re-parametrization invariant as long as the re-parametrizations  are mediated by the intermediate parameter given by the proper time functional $s[x]$: if the proper time $s[x]$ is first evaluated by using an arbitrary parameter $\lambda$, then $\tau[x]$ is indeed  invariant under re-parametrization of  $\lambda$. This is not the usual notion of re-parametrization invariance, but it is good as long as the definition of $\tau[x]$ is given in terms of $s[x]$.

The {\it reality condition} $\tau[x]\in \mathbb{R}$ implies that the proper acceleration respect to the limit metric $g_0$ must be bounded in the following sense,
\begin{proposicion}
In a spacetime $(M_4, g)$ where $g$ is a higher order jet metric given by the expression \eqref{maximalaccelerationmetric}, for any time-like curve with $g(x',x')=\,-1$, the reality of $\tau[x]$ implies the uniform bound on the proper acceleration,
\begin{align}
g_0(D_{x'}x',D_{x'}x')\,\leq \,A^2_{\mathrm{max}}.
\label{bound in the acceleration}
\end{align}
\label{proposicion bound acceleration}
\end{proposicion}
\begin{proof}
For a timelike curve $g(x',x')<0$ and by the expression \eqref{maximalaccelerationmetric0}, then we have that the condition \eqref{bound in the acceleration} holds good.
\end{proof}
\begin{comentario}
The choice of the proper time of $g_0$ as parameter is of relevance because of the following argument: if initially we have a time parameter such that $g_0(x',x')=\,-1$, a change of the time parameter $t\mapsto \kappa$ could make the normalization condition $g_0(x^*,x^*)=\,-1$ not valid.
\end{comentario}
\subsection*{Minimal length of proper time}\label{Minimal lapse of proper time}
The following interpretation of the maximal acceleration is in order. In relativistic theories, namely, theories where for propagation in free space of particles and fields there is an uniform upper bound for local speed and that such a bound is the speed of light in vacuum, the existence of a maximal proper acceleration $A_{\mathrm{max}}$ is equivalent to the existence of a {\it minimal proper time lapse} $\delta{\tau}$ such that $A_{\mathrm{max}}=\,\frac{c}{\delta \tau}$, where $c$ is the speed of light in vacuum. An argument for this property was provided by P. Caldirola in the context of a classical model for the electron \cite{Caldirola1956}. We here extend the argument to any fundamental interaction compatible with an uniform maximal acceleration.

 Because the property described as the thesis in proposition \ref{proposicion bound acceleration}, we call the structures given by \eqref{maximalaccelerationmetric0} or by the expression \eqref{maximalaccelerationmetric}, {\it spacetime metrics of maximal acceleration} or simply {\it metrics of maximal acceleration}. Remarkably, in a spacetime with a metric of maximal acceleration, the upper bound condition \eqref{bound in the acceleration} must hold for any time-like world line $x:I\to M_4$, independently from the association of the time-like curves with physical world lines or when such association is absent.

 Also, note that the condition of maximal acceleration implies a bound on the physical domain of $\epsilon$:
 \begin{corolario}
 In a spacetime with a metric of maximal acceleration, the maximal domain of definition of $\epsilon:\,^2x\to \mathbb{R}$ is given by the interval $[0,1]$.
 \end{corolario}
 Note that in principle there is no physical reason to preclude {\it world lines of maximal acceleration}, $\epsilon =1$, similarly as in relativistic theories, there are {\it world lines of maximal local speed}. In fact, from \eqref{maximalaccelerationmetric} we have that maximal acceleration curves, characterized by
 \begin{align}
 \epsilon (\,^2x)=\,\frac{g_0(D_{x'}x',D_{x'}x')}{A^2_{\mathrm{max}}}=\,1
 \end{align}
are null curves of $g$. Furthermore, since for curves of maximal acceleration
\begin{align*}
g_{^2x}(x',x')=\,g_0(x',x')+1,
\end{align*}
if $g_{^2x}(x',x')=\,g (x',x')=0,$ then $^2x:I\to M_4$ are time like curves of $g_0$ that are parameterized by the proper time of $g_0$.

The metric of maximal acceleration \eqref{maximalaccelerationmetric} has been obtained in the framework of higher order jet metrics under very general assumptions. From one side, we have used assumptions of physical nature, like the existence of the term $\xi$ violating Einstein clock hypothesis, the assumption that physical rays are described by causal curves of $g$ and the existence of the limit metric $g_0$ as a good representation of the spacetime geometric arena when the effects of the acceleration are negligible. On  the other hand, there is the technical hypothesis as the analytical dependence of $\xi$ in the parameter $A^2_{\mathrm{max}}$ and the properties of constancy and non-negativeness of the factors $\{\xi_n,\,n=1,2...\}$, that although motivated on formal logical grounds, maybe could be proved in the future starting from a broader and clear framework and principles.

\section{Jacobi equation and bounds on the Lorentzian curvature}
Given a timelike geodesic $X:I\to M_4$ of the Levi-Civita connection $D$ of the Lorentzian metric $g_0$, the Jacobi equation is the linear equation
\begin{align}
D_{X'}D_{X'} J\,+ R(X',J)\cdot X'=0,
\label{Jacobi equation}
\end{align}
where $R(V,W)$ is the curvature endomorphisms of $D$ associated to the vector fields  $V,W$  along $X:I\to M_4$. The second covariant derivative $D_{X'}D_{X'} J$ admits an heuristic interpretation as {\it relative acceleration} between nearby geodesics of the congruence associated to the Jacobi field  \cite{Hawking Ellis 1973}. As such, the relative acceleration
\begin{align*}
a_r:=\,D_{X'}D_{X'} J:I\to M_4,
\end{align*}
is a vector field along $X:I\to M_4$.

Although $a_r$ is spacelike, it is not obvious that it can be interpreted as the proper acceleration field of a timelike curve $z:\tilde{I}\to M_4$, which is the type of acceleration bounded by a metric of maximal acceleration \eqref{maximalaccelerationmetric0}. However, the following results show that a related interpretation is possible. In particular, each of the values $a_r(t)$ can be interpreted as the initial proper acceleration of a timelike curve. Hence the set $\{a_r(t)\}$ is interpreted as the initial proper acceleration of an aggregate of timelike curves. One can then extend the applicability of the bound \eqref{bound in the acceleration} to the relative acceleration $a_r$ associated to a Jacobi field.

\begin{definicion}
Let $(M_4, g)$ be a spacetime with a metric of maximal acceleration.
A point $p\in\,M_4$ will be called regular if there is a compact set $K\subset\,M_4$ containing $p$ such that the Riemannian curvature of $g_0$ has bound components on $K$. The aggregate of all regular points of $M_4$ is denoted by $M_4(reg)\subset \,M_4$.
\end{definicion}
We assume than around each point $p\in\,M_4(reg)$ we can take the usual derivative and differential operations on forms and tensors. This is the case when $M_4(reg)$ is a manifold.

\begin{lema}
Let $(M_4,g)$ be a spacetime with a metric of maximal acceleration. Let $X:I\to M_4$ be a timelike geodesic of $g_0$ parameterized by proper time, $g_0(X',X')=-1$, and let $p=X(0)$ be a regular point.  Then there is a Jacobi field ${J}:I\to M_4$ along $X:I\to M_4$ and a timelike curve $z:\tilde{I}\to M_4$ with initial conditions $z(0)=X(0)$, $\dot{z}(0)=\,X'(0)$ such that
\begin{itemize}
\item The condition $D_{\dot{z}}\dot{z}|_0 =\,D_{X'}D_{X'}J |_0$ holds good,

\item The uniform bound
\begin{align}
g_0(D_{{X}'}D_{{X}'} {J}|_{t=0},D_{{X}'}D_{{X}'} {J}|_{t=0})\leq \,A^2_{\mathrm{max}}
\label{bound for relative accelerations}
\end{align}
holds good.
\end{itemize}
\label{lema on bound of the relative acceleration}
\end{lema}

\begin{proof}
Let us consider a geodesic  $X:I \to M_4$  such that $p=X(0)$ is regular and the geodesic is parameterized by the proper time, $g_0(X'(t),X'(t))=\,-1$. We will consider differential equations of the form,
\begin{align}
D_{\dot{z}}\dot{z}(s)=\,\mathcal{F}_{(p,J(0))}(z(s)),
\label{transversal world lines}
\end{align}
where the dot-derivatives are taken respect to proper parameter of $g_0$ along $z:\tilde{I}\to M_4$. The initial values for the condition \eqref{transversal world lines} are
\begin{align*}
(z(0),\dot{z}(0))=\,(X(0),{X}'(0)).
\end{align*}
Therefore, the constrain
\begin{align*}
g_0(\dot{z}(0),\dot{z}(0))=\,g_0(X'(0),X'(0))=\,-1
\end{align*}
must hold good.
As a consequence and since the model that we will propose for $z:\tilde{I}\to M_4$ is a continuous model, we can assume that the smooth  curve $z:\tilde{I}\to M_4$ is timelike. On the other hand, if $D_{\dot{z}}\,\dot{z}$ is interpreted as a proper acceleration, $\mathcal{F}_{(p,J(0))}(z(s))$ needs to be spacelike vector field along $z:\tilde{I}\to M_4$. Hence the orthogonality condition
\begin{align}
g_0\left(\dot{z}, \mathcal{F}_{(p,J(0))}(z)\right)=0
\label{orthogonal condition for the fiducial field}
\end{align}
must hold.

Let us consider the ansatz
\begin{align}
 \mathcal{F}_{(p,J(0))}(z(s))=\,R(\hat{J}(s),\dot{z}(s))\cdot \dot{z}(s).
 \label{tidal force}
\end{align}
 Then the orthogonal initial conditions \eqref{orthogonal condition for the fiducial field} holds good.
 Furthermore, we need to impose a constrain to determine $\hat{J}(s)$ along $z:\tilde{I}\to M_4$. For this purpose we assume that
 \begin{align}
 g_0(\dot{z}(s),\hat{J}(s))=\,0.
 \label{orthogonality condition for J z}
 \end{align}
 Taking the derivative of this constrain along $z:I\to M_4$ it follows that
\begin{align*}
 0 & =\,\frac{d}{ds}\left(\,g_0\left(\dot{z}(s),\hat{J}(s)\right)\right)\\
 & =\,g_0\left(D_{\dot{z}}\,\dot{z}(s),\hat{J}(s)\right)+\,g_0\left(\dot{z}(s),D_{\dot{z}}\hat{J}(s)\right)\\
& =\,g_0\left(R( \hat{J}(s),\dot{z}(s))\cdot \dot{z}(s), \hat{J}(s)\right)+\,g_0\left(\dot{z}(s),D_{\dot{z}}\hat{J}(s)\right)\\
& =\,-g_0\left(R( \hat{J}(s),\dot{z}(s))\cdot \hat{J}(s),\dot{z}(s)\right)+\,g_0\left(\dot{z}(s),D_{\dot{z}}\hat{J}(s)\right).
\end{align*}
 This constrain can be satisfied if we declare the vector field $\hat{J}(s)$ is a solution of the differential equation
\begin{align}
D_{\dot{z}}\hat{J}(s)\,-R( \hat{J}(s),\dot{z}(s))\cdot \hat{J}(s)=0.
\label{extension of the Jacobi field}
\end{align}
This equation is subjected to the initial condition
\begin{align}
\hat{J}(0)=\,J(0).
\label{initial condition for hat J}
\end{align}
If $J:I\to M_4$ is a Jacobi field along $X:I\to M_4$, then $D_{\dot{X}}D_{\dot{X}}J|_{t=0}$ is spacelike and the initial condition \eqref{initial condition for hat J} ensures the orthogonality condition \eqref{orthogonality condition for J z} at the initial time $s=0$ and the equation \eqref{extension of the Jacobi field} ensures that such condition is preserved along $z:\tilde{I}\to M_4$. From now on, we assume $J(0)$ is the initial value of the Jacobi field $J:I\to M_4$ along the central geodesic $X:I\to M_4$.

In a compact way, the ordinary differential equations that determine the above construction are
\begin{align}
\begin{cases}
& D_{\dot{z}}\dot{z}(s)=\,R(\hat{J}(s),\dot{z}(s))\cdot \dot{z}(s), \\
& D_{\dot{z}}\hat{J}(s)=\,R(\hat{J}(s),\dot{z}(s))\cdot \hat{J}(s)
\end{cases}
\label{compact way differential equations}
\end{align}
Given the curvature endormorphism $R( \hat{J}(s),\dot{z}(s))$, the system of differential equations \eqref{compact way differential equations} together with the initial conditions
\begin{align}
\begin{cases}
& (z(0),\dot{z}(0))=\,(X(0),{X}'(0))\\
&\hat{J}(0)=\,J(0).
\end{cases}
\label{initial conditions for z hatJ}
\end{align}
  determine a timelike curve whose initial acceleration are the specified one: $z(0)=X(0)$, $\dot{z}(0)=X'(0)$ and the initial proper acceleration $D_{\dot{z}}\dot{z}|_{s=0}$ is the relative acceleration $a_r=\,D_{X'}D_{X'} \,J|_{t=0}$ of the central geodesic $X:I\to M_4$ at $t=0$.

The system of equations \eqref{compact way differential equations} can be expressed in local components by the equations
\begin{align}
 \begin{cases}
&\ddot{z}^\mu(s)=\,-\Gamma^\mu\,_{\nu\rho}(z(s))\,\dot{z}^\nu(s)\,\dot{x}^\rho(s)\,+\left(R(\hat{J}(s),\dot{z}(s))\cdot{\dot{z}}(s)\right)^\mu,\\
&\dot{\hat{J}}^\mu(s) =\,-\Gamma^\mu\,_{\nu\rho}(z(s))\,\dot{z}^\mu\,\hat{J}^\rho(s)+\left(R(\hat{J}(s),\dot{z}(s))\cdot{\hat{J}}(s)\right)^\mu,\\
& \mu,\nu,\rho =0,1,2,3.
\end{cases}
\label{coordinate for of the equations z hatj}
\end{align}
subjected to the initial conditions \eqref{initial conditions for z hatJ}.
Since the right hand sides of the  equations \eqref{coordinate for of the equations z hatj} are $C^1$-smooth functions on $(s,z,\dot{z},\hat{J})$, Peano existence theorem (see for instance \cite{Hartman}, pg. 10) can be applied as follows. First, in order to simplify the expressions, we can choose Fermi coordinates along $z:\tilde{I}\to M_4$ such that $\Gamma^\mu\,_{\nu\rho}(z(s))=0$. Second, let us note that the curvature endomorphism $R(\hat{J}(s),\dot{z}(s)):T_{z(s)} M_4\to T_{z(s)}M_4$ with $s\in\,\tilde{I}=\,[0,a]$, acts on the vector fields $\dot{z}(s)$ and $\dot{J}(s)$ in a $C^2$-smoothly way.  Let us denote by $u(s)$ the triple $\left(z(s),\dot{z}(s),\hat{J}(s)\right)$ and consider the bound
\begin{align}
 \nonumber \quad \quad \quad &\max\{  |g_0(R(\hat{J}(s),\dot{z}(s))\cdot \dot{z}(s),R(\hat{J}(s),\dot{z}(s))\cdot \dot{z}(s))|,\\
& |g_0(R(\hat{J}(s),\dot{z}(s))\cdot \hat{J}(s),R(\hat{J}(s),\dot{z}(s))\cdot \hat{J}(s))|\,s\in\,[0,a]\}<\,M^2(u(0)),
\label{Peano relation}
\end{align}
where $M(u(0))$ is a constant depending on the initial conditions $u(0)=(z(0),\dot{z}(0),\hat{J}(0))$. Third,
in a compact domain $K\subset M_4$ containing the image $z([0,a])$, we can define a distance function $d_K:K\times K \to \mathbb{R}$. The compact domain $K$ and the distance $d_K$ could depend upon de local coordinate system where we are applying Peano existence theorem. This distance function can be extended to define the distance between $u(0)$ and $u(s)$, that we denote by $d(u(s),u(0))$.
In these conditions, Peano existence theorem \cite{Hartman} can be applied. It implies the existence of the solution for the system \eqref{coordinate for of the equations z hatj} for the given initial conditions in the interval of parameter domain $[0,\alpha(X(0), J(0))]\subset \,[0,a]$, where in our case the constant $\alpha(X(0), J(0))$ is given by the expression
\begin{align*}
\alpha(X(0),J(0))=\,\min \left( a,\frac{\max \{d(u(s),u(0))\,s\in [0,a]\}}{M(u(0))}\right).
\end{align*}

The relative acceleration $a_r=\,D_{X'}D_{X'} J|_{t=0}$ is interpreted as the initial proper acceleration of a curve $z:[0,\alpha]\to M_4$. In order to apply the relation \eqref{bound in the acceleration}, it is necessary to assure that the Jacobi field $\hat{J}|_{t=0}=J(0)$ is compatible with two constraints:
\begin{itemize}
\item The parametrization condition $g_0(\dot{z}, \dot{z})=\,-1$, which is direct, since the parameter $s\in \,[0,\alpha]$ is independent of the proper parameter $t\in I$.

\item By the hypothesis of the Lemma, specifically, the assumption that the geometry of the spacetime is of maximal acceleration, there must be a minimal lapse of proper time $\delta \tau$.
\end{itemize}
Hence the value of the initial Jacobi field must be such that $\alpha(X(0),J(0))\geq \delta \tau$. This can be achieved easily by a change in the initial conditions $\hat{J}(0)\mapsto \lambda\,\hat{J}(0)$, such that for $\lambda$ small enough, it makes the ratio
\begin{align*}
 \frac{\max \{d(u_\lambda(s),u_\lambda(0))\,s\in [0,a]\}}{M(u_\lambda(0))}
  \end{align*}
  larger than $\delta \tau$, where in this expression $u_\lambda(s)$ (resp. $u_\lambda(0)$) are the solutions (resp. initial conditions) of the problem \eqref{coordinate for of the equations z hatj} with the initial conditions $u_\lambda(0)=\,(z(0),\dot{z}(0),\lambda \,J(0))$. This constraint implies a limit on  the size of the Jacobi equation that we can use, despite the fact that the Jacobi equation \eqref{Jacobi equation} is a linear equation.
  Indeed, under the change of initial condition $J(0)\mapsto \lambda\,J(0)$ for asymptotically large $|\lambda|$ we have that
\begin{align*}
\alpha(X(0),J(0))&\mapsto  \min \left( a,\frac{\max \{d(u_\lambda(s),u_\lambda(0))\,s\in [0,a]\}}{\,M(u_\lambda(0))}\right)\\
& \sim \,\frac{1}{|\lambda |^2}\,\frac{\max \{d(u_\lambda(s),u_\lambda(0))\,s\in [0,a]\}}{\,M(u(0))}\\
& \sim \,\frac{1}{|\lambda |}\,\frac{\max \{d(u(s),u(0))\,s\in [0,a]\}}{\,M(u(0))}< \,\delta \tau
\end{align*}
 for  large enough $|\lambda|$, in contradiction with the hypothesis that $\delta\tau$ is the minimal lapse of proper time. Therefore, not all the Jacobi vector fields $J(0)$ are compatible with the existence of a minimal value for the parameter $a=\delta \tau$.

From the above considerations, the relation \eqref{bound for relative accelerations} follows for Jacobi field $J$ with small enough norm. The result is general covariant, even if at some point of the proof we have made use of Fermi coordinates.

\end{proof}

The application of Lemma \ref{lema on bound of the relative acceleration}  at each $p\in M_4(reg)$ leads to the following
\begin{proposicion}Consider the spacetime with a metric of maximal acceleration $(M_4,g)$. Then there is a Jacobi field $J:I\to M_4$ along the timelike geodesic $X:I\to M_4(reg)$ compatible with the condition $g_0(X',X')=-1$ such that  the uniform bound
\begin{align}
g_0 (R(X',J)\cdot X',\,  R(X',J)\cdot X')\leq A^2_{\mathrm{max}},
\label{master relation for curvature bounds}
\end{align}
holds good along the geodesic,
where $R(\cdot,\cdot)$ is the curvature of the Lorentzian metric $g_0$.
\label{proposicion on bound in curvature}
\end{proposicion}

Let us consider a particular initial condition $(X(0), X'_0(0))$, with $p=\,X(0)\,\in M_4(reg)\,\subset M_4$ and $X'_0(0)\in\,T_p M_4$ is a timelike tangent vector such that $g_0(X'(0),X'(0))=\,-1$. Let $\{X'(0),J_b,\,b=1,2,3\}$ be a tangent basis for $T_{p}M_4$ and assume that $X'_0(0)$ is a timelike vector and $\{J_b(0)\,b=1,2,3\}$ are supplementary tangent vectors at $X(0)$, that without loss of generality, satisfy the relations
\begin{align*}
&g_0(X'_0(0), X'_0(0))=\,-1, \quad g_0(J_{a}(0),X'_0(0))=0,\quad  g_0(J_{a}(0),J_{b}(0))=\,\delta_{ab},\\
& a,b=1,2,3.
\end{align*}
Let $J:I\to M_4$ be a Jacobi field satisfying the conditions of proposition \eqref{proposicion on bound in curvature}. In terms of the above basis, we have
\begin{align*}
R(J(0),X'(0))\cdot X' & =\,\sum^3_{b=1}\sum^3_{\mu=0}\, R^{b}\,_{0\mu 0}\,X'(0)^0\,J^\mu(0)\,J_b(0)+ \,\sum^{3}_{\mu=0}\,R^0\,_{0\mu0}\,X'(0)^0\,J^\mu(0)\,X'(0)\\
&=\,\sum^3_{b=1}\sum^3_{\mu=0}\, R^{b}\,_{0\mu 0}\,J^\mu(0)\,J_b(0)+ \,\sum^{3}_{\mu=0}\,R^0\,_{0\mu0}\,J^\mu(0)\,X'(0)\\
&=\,\sum^3_{b=1}\sum^3_{\mu=0}\, R^{b}\,_{0\mu 0}\,J^\mu(0)\,J_b(0),
\end{align*}
since $X'(0)^0 =1$ in the above basis and $R^0\,_{0\mu 0}=0$ because the symmetries of the curvature endomorphism of the Riemann curvature.
Applying the orthonormal relations, we obtain for every admissible $J(0)$ compatible with the minimal proper time lapse hypothesis the relation
\begin{align*}
& g_0(R(J(0),X'(0))\cdot X', R(J(0),X'(0))\cdot X')\\
&=\,\sum^3_{a=1}\,\sum^{3}_{\mu,\nu=0}\,R^a\,_{0\mu0}\,R_{a0\nu 0}\, J^\mu(0)\,J^\nu(0)\\
 &=\,\sum^{3}_{\rho,\mu,\nu=0}\,R^\rho\,_{0\mu0}\,R_{\rho 0\nu 0}\, J^\mu(0)\,J^\nu(0)\leq \,A^2_{\mathrm{max}}.
\end{align*}
Every spacetime direction $U$ at $T_p M_4$ has associated an unique admissible spacelike vector $J_{\mathrm{max}}=\,\lambda_{\mathrm{max}}\,U(0)$ such that one can apply \ref{proposicion on bound in curvature}, leading to the relation
\begin{align*}
\sum^3_{a=1}\,\sum^{3}_{\mu,\nu=0}\,R^a\,_{0\mu0}\,R_{a0\nu 0}\, J^\mu_{\mathrm{max}}(0)\,J^\nu_{\mathrm{max}}(0)\leq A^2_{\mathrm{max}},
\end{align*}
where $J_{\mathrm{max}}(0)$ is the spacelike vector proportional to $J(0)$ compatible with the minimum time lapse $\delta \tau$ and of maximal modulus with $g_0$. Then we have the following
\begin{proposicion}
In a spacetime geometry of maximal acceleration $(M_4,g)$ and in the basis $\{X'(0), J_a,\,a=1,2,3\}$,  the Riemannian curvature  components are uniformly bounded in the sense that
\begin{align}
\sum^{3}_{\rho,\mu,\nu=0}\,R^\rho\,_{0\mu0}\,R_{\rho 0\nu 0}\, J^\mu_{\mathrm{max}}(0)\,J^\nu_{\mathrm{max}}(0)\leq A^2_{\mathrm{max}}\,\quad \mu,\,\nu,\rho=0,1,2,3
\label{bounds on curvature due to maximal acceleration}
\end{align}
 at any regular point $p\in\,M_4(reg)\subset \,M_4$ and maximal admissible spacelike vector $J_{\mathrm{max}}\,\in T_{p}M_4$.
 \label{proposicion on bound on curvatures}
\end{proposicion}
As a consequence the existence of a minimal proper time lapse $\delta \tau$ and proposition \ref{proposicion on bound on curvatures}, we have the following
\begin{teorema}
The components of the curvature endomorphism $R^\rho\,_{0\mu0}$ and the components of the curvature tensor $R_{\rho 0\nu 0}$ remain finite for any regular point $p\in\,M_4(reg)$.
\label{upper bound in the curvature theorem}
\end{teorema}
\begin{proof}
The tangent vector  $J_{\mathrm{max}}\in\,T_p M_4$ is different from zero and it has finite norm. This follows from the existence of a minimum lapse of proper time $\delta \tau$. Indeed one has
\begin{align*}
\delta \tau\, &\leq \alpha(X(0),J(0)) \leq \,\frac{\max \{d(u(s),u(0))\,s\in [0,a]\}}{\,M(u(0))}.
\end{align*}
But for $\|J_\mathrm{max}\|$ very small, $M(u(0))$ can be approximated by $A_{\mathrm{max}}$, because the form of the relation \eqref{Peano relation} that involves a linear term and a quadratic term in $\hat{J}$  and because the application of proposition \ref{proposicion on bound on curvatures}. Hence
\begin{align*}
A_{\mathrm{max}}\,\delta \tau\,  \leq \,\max\{d(u(s),u(0))\,s\in [0,a]\}.
\end{align*}
The distance function $d(u(s),u(0))$ is a smooth function on the solutions, that at the same time are smooth functions on the parameter $\|J_0\|$ (see for instance \cite{Hartman}, pg. 101) and therefore, on $\|J_\mathrm{max}\|$. The distance function can be expressed as a Taylor expansion up to first order in $\|J_0\|$ (the reminder term being of second order). Hence the above bound fix a lower bound for $\|J_0\|$ and hence for $\|J_\mathrm{max}\|$.
\end{proof}

Let us synthesise the content of this section as follows. First, the construction of the world line $z:\tilde{I}\to M_4$ shows how the relative acceleration $a_r$ can be interpreted as a proper acceleration. Hence we can apply the bound in the acceleration \eqref{bound in the acceleration}, a bound which is consistent with the linearity of the Jacobi equation \eqref{Jacobi equation}. But what it is further restrictive is that a maximal proper acceleration implies the existence of a minimal lapse of time, as discussed in \cite{Caianiello1984} and also in the sub-section \ref{Minimal lapse of proper time}. This leads to upper bounds on the curvature as in theorem \ref{upper bound in the curvature theorem}.

\section{The equivalence principle for spacetimes of maximal acceleration}
{\it Einstein equivalence principle} is on the basis of our current theories for the gravitational interaction. Following Thorne, Lee and Lightman \cite{ThorneLeeLightman}, the principle can be stated as follows:
\bigskip
\\
{\bf Einstein equivalence principle}. The following postulates hold good:
\begin{enumerate}
\item The weak equivalence principle or {\it universality of free falling} holds: {\it if an un-charged test body is placed at an initial event in spacetime and is given an initial velocity there, then its subsequent world line will be independent of its internal structure and composition.}

\item {\it 1. The outcome of any local, non-gravitational test experiment is independent of where and when in the universe is performed and 2. It is independent of the velocity of the free falling experimental apparatus where the experiment is realized.}

\end{enumerate}
The first part of the postulate in Einstein equivalence principle is full-filed if there is a connection $\widetilde{D}$ whose auto-parallel curves are in one to one correspondence with the world lines of free falling bodies. Then the  connection $\widetilde{D}$ must be independent of the composition and structure of the free falling test bodies. This form of the postulate is not difficult to accomplish in the category of Finslerian spacetimes, as discussed for example in \cite{Ricardo2017}. In such a framework, the Lorentzian models of spacetime are a sub-category and is where general relativity finds its mathematical formulation, but also note that in the Finslerian category there are connections violating the postulate. Hence the weak equivalence principle is already a restriction on the mathematical formalism.

The second postulate in Einstein equivalence principle presupposes the existence of smooth free falling local coordinate systems where the experimental set up can be constructed and where all the gravitational effects can locally be eliminated.  These free falling coordinates are usually mathematically implemented as the Fermi coordinate system along $X:\tilde{I}\to M_4$ representing the world line of the laboratory system. In such free falling coordinate systems, the outcomes of any experiment showing deviation from the free evolution of the studied physical systems indicates the presence of a non-gravitational field. This way of identifying non-gravitational fields fails when the dynamical effects of the connection cannot be eliminated locally in the free falling local coordinate system.

The second postulate in Einstein equivalence principle also introduces a strong constrain on the possible physical laws: the outcome of non-gravitational experiments must be independent of the state of motion of the free falling local coordinate systems. This part of the postulates implies that in local free falling coordinate frames, the laws describing physical phenomena are consistent with a theory of relativity. For example, the law could be consistent with Galilean relativity or could be consistent with Einstein special relativity.
Examples of gravitational theories obeying Einstein equivalence principle are Newton-Cartan theory and general relativity \cite{ThorneLeeLightman}.

If the connection $\widetilde{D}$ defined by free fall motion is either an affine connection defined over the spacetime $M_4$ or it is determined in a direct way by an affine connection on $M_4$, then the first postulate and the first part of the second postulate in Einstein equivalence principle can be implemented geometrically and rather independently of the second half of the second postulate. The second part formulated in the second postulate is full-filled if the spacetime geometry is consistent with a principle of relativity, for instance, either consistent with Galileo or with Einstein special relativity. These type geometries are to be found in Finsler spacetimes geometries, from which Lorentzian spacetime geometry is the standard case.

Theories of the gravitational interaction are based upon the principle that the effects due to gravity solely can be locally suppressed in local free falling coordinate systems.
This suggests an alternative formulation of the equivalence principle, that although is not as strong in the constrains that the Einstein principle imposes, it still captures this essential characterization of gravity. Modified in the following way, the new version of the equivalence principle can be putted to work under similar purposes than the original Einstein equivalence principle, but the new version can be accommodated to spacetimes with a maximal proper acceleration, where neither the Galilean relativity or special relativity hold. In this context, we propose to consider the following:
\bigskip
\\
{\bf New version of the equivalence principle}. The following two conditions hold good:
\begin{enumerate}
\item The weak equivalence principle or {\it universality of free falling} holds: {\it if an un-charged test body is placed at an initial event in spacetime and is given an initial velocity there, then its subsequent world line will be independent of its internal structure and composition.}

\item {\it Existence of smooth free falling local coordinate systems, where the experimental set up is constructed and where gravitational effects integrated in the connection can be locally eliminated.}

\end{enumerate}
In this new version of the equivalence principle, gravitational interactions are characterized exactly as the ones that for  non-charged point test particle, the given interaction can be {\it switch-off locally} in local free falling coordinate systems. By {\it switch-off} locally we mean that the effect on the dynamics is trivialized for world lines close enough to the given free fall coordinate system. A detailed specification of the concept is discussed in \cite{ThorneLeeLightman}, a treatment of this notion that we adopt as valid for the present work.

This provides a strong restriction on the possible geometric models, as we will see in the next section. However, let us note that the principle here considered can be extended in the form of a generalization or modification of the Einstein equivalence principle, a necessary step towards field equations involving spacetimes with metrics of maximal acceleration. Although this could imply phenomenological deviations from Einstein equivalence principle, we will not consider this problem in this work, leaving it for future research. We will here concentrate on the formal implementation of the new formulation of the equivalence principle.

\section{A geometric framework for spacetimes with a metric of maximal acceleration}
We consider now a geometric framework for spacetimes with maximal acceleration where the two conditions of the new version of the equivalence principle are satisfied.
The connection $\widetilde{D}$ will be a connection on a fibre bundle $\pi:\mathcal{E}\to M_4$. For instance, in the case of models of gravitational interaction based upon the theory of Finsler spacetimes \cite{Beem1970, GallegoPiccioneVitorio:2012}, the connections are defined on vector bundles over the slit tangent bundle $TM\setminus\{0\}\to M_4$. Similarly, in the case of spacetimes of maximal acceleration, $\widetilde{D}$ will be a connection defined on vector bundles over the $2$-jet bundle $J^2M_4$.

\subsection{The equivalence principle and the constrains on the connection}
Let us consider the {\it pull-back bundle} $p_1:\pi^*_2 TM_4\to J^2M_4$, defined by the commutative diagram
\begin{align*}
\xymatrix{\pi^*_2TM_4 \ar[d]_{p_1} \ar[r]^{p_2} &
 TM_4 \ar[d]^{\pi}\\
J^2M_4 \ar[r]^{\pi_2} & M_4.}
\end{align*}
$\pi^*_2J^2M_4$ is a vector bundle, where the fiber $p^{-1}_1(u)\subset \,\pi^*_2TM_4$ with  $u\in J^2M_4$, is diffeomorphic to the vector space $T_{\pi_2(u)}M_4=\pi^{-1}(x)$ for each $x=\pi_2(u)\in\,M_4$. Note that since $\pi_2:J^2M_4 \to M_4$ is a projection onto $M_4$, then
\begin{align*}
rank(d\pi_2)=\,\dim(M_4)=\,4.
\end{align*}

The connections that we will consider  on the vector bundle  $p_1:\pi^*_2 TM_4\to J^2M_4$ are inspired by analogous Finslerian  constructions \cite{BaoChernShen}. In a similar way as for Finslerian spacetime models of gravity  \cite{Ricardo2017}, the new form of the equivalence principle as stated in the previous section imposes  non-trivial constrains on the connection $\widetilde{D}$ defined on the vector bundle $\pi^*_2 TM_4$. In particular, the connection $\widetilde{D}$ must be compatible with the existence of smooth free falling local coordinate systems. A candidate for $\widetilde{D}$ with such property is the pull-back connection $\pi^*_2 D$ of the Levi-Civita connection of the Lorentzian metric $g_0=\,\lim_{A_{\mathrm{max}}\to +\infty}$, a connection which is defined by the relation
\begin{align}
\left(\pi^*_2 D\right)_X\,\pi^*_2 Z =\,\pi^*_2\left(D_{d\pi_2(X)}\,Z\right),\quad Z\in\,\Gamma TJ^2M_4, \,Z\in\,\Gamma \,TM_4.
\label{induced connection in pullback}
\end{align}

The introduction of free falling coordinate systems associated to the connection $\widetilde{D}$ requires first the discussion of several notions and results. Let us first introduce the following notion of auto-parallel condition in $\pi^*_2TM_4$,
\begin{definicion}
An auto-parallel curve in $J^2M_4$ is a curve $\gamma:[a,b]\to J^2M_4$ whose tangent vector field $X$ is such that $d\pi_2(X)=Z$ satisfies the auto-parallel condition $D_Z \,Z=0$ of the Levi-Civita connection $D$ of the Lorentzian metric $g_0$.
\label{autoparallel condition in second jet}
\end{definicion}
$\pi_2:J^2M_4\to M_4$ is a canonical projection and $D$ is the Levi-Civita connection on $M_4$. Hence we have the following
\begin{proposicion}
The following properties of the connection $\pi^*_2 D$ of the bundle $\pi^*_2 TM_4$ hold good:
\begin{enumerate}
\item The connection $\pi^*_2 D$ is a linear connection.

\item The connection coefficients of $\pi^*_2D$ live on the spacetime manifold $M_4$.

\item The auto-parallel condition \eqref{autoparallel condition in second jet} is equivalent to the auto-parallel condition of the Levi-Civita connection $D$.

\end{enumerate}
\label{properties of the induced connection}
\end{proposicion}
\begin{proof}
The first property is direct from the construction of the pull-back connection.

The second property is proved using local natural coordinate systems $(\,^2x^A)$ in $TJ^2M_4$ and the definition \eqref{induced connection in pullback} of $\pi^*_2 D$.

The third property is direct from the definition of $\pi^*_2 D$ and the definition \ref{autoparallel condition in second jet} of the pull-back connection. Indeed,  if $D_Z\,Z=0$, then
\begin{align*}
\pi^*_2 D_X\,\pi^*_2 Z =\,\pi^*(D_{d\pi_2(X)})\,Z= \,\pi^*_2 \,\left(D_Z Z\right)=0.
\end{align*}
 Conversely,  since $\pi^*_2 Z_1 =\,\pi^*_2 Z_2$ iff $Z_1=\,Z_2$, then $\pi^*_2 D_X\,\pi^*_2 Z=0$ implies $D_Z \,Z=0$.
\end{proof}

Clearly, we can adopt as connection $\widetilde{D}$ in our theory of spacetimes with metrics of maximal acceleration the connection $\pi^*_2D$,
\begin{align}
\widetilde{D}:=\,\pi^*_2D.
\end{align}
To understand how this connection $\pi^*_2D$ admits free falling coordinate systems in the form of Fermi-like coordinates, we need to investigate some of its basic geometric properties.
\subsection{Properties of the connection $\widetilde{D}$}
We study now some basic properties of the connection $\pi^*_2 D$.
\subsection*{The Torsion tensor of the pull-back connection $\pi^*_2 D$}
There is no defined a conventional torsion tensor for a given connection $\nabla$ in $\pi^*_2J^2M_4$. However, a  close analogous tensor field along $\pi_2$  that one can consider is defined as follows. If $d\pi_2 (X_i)=\,Z_i$ are local vector fields on $U\subset \,M_4$, then
\begin{align}
\widetilde{T}_\nabla(X_i,X_j):=\,\nabla_{X_i} \,\pi^*_2Z_j-\,\nabla_{X_j} \,\pi^*_2Z_i-\pi^*_2([Z_i,Z_j]).
\end{align}
We can call $\widetilde{T}_\nabla$ the pull-back torsion tensor of the connection $\nabla$.
Then we have the following
\begin{proposicion}
The pull-back torsion tensor of $\pi^*_2D$ is zero,
\label{zero torsion for the induced connection}
\begin{align}
\widetilde{T}(X_i,X_j)=\,0,\quad\quad i,j=1,2,3,4.
\label{zero torsion for the pull-back condition}
\end{align}
\end{proposicion}
\begin{proof}
For $\pi^*_2D$ we have that
\begin{align*}
\widetilde{T}(X_i,X_j)& =\,\pi^*_2 ( D_{X_i} \,\pi^*_2Z_j)-\,\pi^*_2( D_{X_j} \,\pi^*_2Z_i)-\pi^*_2([Z_i,Z_j])\\
& =\,\pi^*_2 ( D_{Z_i} \,Z_j-\, D_{Z_j} \,Z_i-[Z_i,Z_j])=0,
\end{align*}
since for the Levi-Civita connection $D_{Z_i} \,Z_j-\, D_{Z_j} \,Z_i-[Z_i,Z_j]=\,0$.
\end{proof}
\subsection*{Fermi coordinates for the pull-back connection $\pi^*_2 D$}
The theory of connections developed above accommodates the new version of the equivalence principle. Proposition \ref{properties of the induced connection} implies that auto-parallel curves of $\pi^*_2D$ are the geodesics of the Levi-Civita connection $D$, full-filling the first condition of the new version of  the equivalence principle. Moreover, the existence of local coordinates where all the gravitational effects encoded in the connection can be eliminated in small enough neighborhood, is realized by the existence of Fermi frames for $D$ along arbitrary geodesics $x:I\to M_4$. Since $rank(d\pi_2)=4$, this implies local frames of $J^2M_4$ along the corresponding lift $^2x:I\to J^2M_4$ where the connection coefficients of the pull-back connection $\pi^*_2D$ are zero.
\subsection*{Curvature and Jacobi fields of the pull-back connection $\pi^*_2 D$}
One can consider a generalization of the Jacobi equation for the connection $\pi^*_2 D$. The curvature endomorphism of $\pi^*_2D$ is defined by the expression
\begin{align}
&\widetilde{R}(X_1,X_2)\cdot \,\Xi:  =\,\left(\pi^*_2D_{X_1}\,\pi^*_2 D_{X_2}\,-\pi^*_2D_{X_2}\,\pi^*_2 D_{X_1}\,-\pi^*_2D_{[X_1,X_2]}\right)\cdot \Xi,\\
\nonumber & X_1,X_2\,\in \,\Gamma TJ^2M_4, \, \Xi\in\,\Gamma \pi^*_2TM_4.
\end{align}
Let us consider $X_1, X_2\in\,\Gamma TJ^2M_4$ such that $[X_1,X_2]=0$ with $d\pi_2 (X_1)=\,Z_1$ and $\pi_2 (X_2)=\,Z_2$. Then $[Z_1,Z_2]=0$ holds good. By the torsion-free property of proposition \ref{zero torsion for the induced connection},
\begin{align*}
\pi^*_2D_{X_1}\,\pi^*_2 D_{X_1}\,\pi^*_2 Z_2 & =\,\pi^*_2D_{X_1}\,\pi^*_2 D_{X_2}\,\pi^*_2 Z_1 \\
& =\,\pi^*_2D_{X_2}\,\pi^*_2 D_{X_1}\,\pi^*_2 Z_1-\,\widetilde{R}(X_2,X_1)\cdot \pi^*_2 Z_1.
\end{align*}
If the auto-parallel condition
\begin{align}
\pi^*_2 D_{X_1}\,\pi^*_2 Z_1 =0
\end{align}
holds good, then we have
\begin{align}
\pi^*_2D_{X_1}\,\pi^*_2 D_{X_1}\,\pi^*_2 Z_2 +\,\widetilde{R}(X_2,X_1)\cdot \pi^*_2 Z_1=0.
\label{Jacobiequation for induced connection}
\end{align}
Equation \eqref{Jacobiequation for induced connection} is the Jacobi equation for the
pull-back connection $\pi^*_2 D$.

The Jacobi fields of $D$ and the Jacobi fields of $\pi^*_2 D$ are related. First note that if $d\pi_2(X_i)=\,Z_i,\,i=1,2$, then from the definition of $\pi^*_2 D$ it follows that
\begin{align}
\pi^*_2D_{X_1}\,\pi^*_2 D_{X_1}\,\pi^*_2 Z_2=\,\pi^*_2D_{X_1}\left(\pi^*_2\left( D_{Z_1}\,Z_2\right)\right)=\,\pi^*_2\,\left( D_{Z_1}\,D_{Z_1}\,Z_2\right)
\label{pullback relation 1}
\end{align}
and as consequence
\begin{align}
\widetilde{R}(X_1,X_2)\cdot X_1=\,\pi^*_2\left( R(Z_1,Z_2)\cdot Z_1\right).
\label{pullback relation 2}
\end{align}
By direct application of these two relations, we can prove the following
\begin{proposicion}
The Jacobi fields of $\pi^*_2 D$ and $D$ are in one to one correspondence.
\end{proposicion}
\label{Relation between Jacobi fields of D and induced connection}
\begin{proof}
If $X_2$ is a Jacobi field of $\pi^*_2D$ along $\gamma:I\to M_4$ with tangent vector field $\gamma' = X_1:I\to TJ^2M_4$ and such that $[X_1,X_2]=\,0$, then by the relations \eqref{pullback relation 1}-\eqref{pullback relation 2}, the Jacobi equation \eqref{Jacobiequation for induced connection} and the definition of auto-parallel condition $\pi^*_2 D X_1\,\pi^*_2 Z_1=\,0$, the field $Z_2 =\,d\pi_2 (X_2):I\to TM_4$ is a Jacobi field $\pi^*_2 Z_2$ along the curve $\pi_2\circ \gamma:I\to TM_4$.

By the same arguments, but reversing the implications, it is shown that for any Jacobi field $Z_2$ along a geodesic of $D$, there is a Jacobi field of $\pi^*_2 D$ along the lift $\gamma=\,^2x$ to the second jet bundle.
\end{proof}
\begin{comentario}
The results of this sub-section still hold if instead of $D$ we consider a torsion-free affine connection $D'$.
\end{comentario}
\subsection{Relation between the metric of maximal acceleration and the pull-back connection $\pi^*_2 D$}
According to the geometric framework discussed above, the properties of $D$ and $\pi^*_2 D$ are intimately related. It has also been convincingly showed that Jacobi fields of $\pi^*_2 D$ are determined by the Jacobi fields of $D$, the Levi-Civita connection of the Lorentzian metric $g_0$. We show now the relation between the metric of maximal acceleration $g$ and the pull-back connection $\pi^*_2 D$.

Instead of considering directly the metric $g=\,(1-\epsilon)\,g_0$, we will work with a {\it  deformed pull-back fiber metric} defined as acting on sections of $p_1:\pi^*TM_4\to J^2M_4$. This deformation is defined by the expression
\begin{align}
\tilde{g}:=\left(1-\epsilon\,\right)\,\pi^*_2g_0,
\label{pull-back metric}
\end{align}
where the pull-back metric $\pi^*_2 g_0 $ is defined by the relation
\begin{align*}
\pi^*_2g_0 (\pi^*_2 S,\pi^*_2 S):=\,g_0(S,S),\quad \, \forall \,S\in \,\Gamma TM_4.
\end{align*}
A direct computation of the covariant derivative of the pull-back metric \eqref{pull-back metric} leads to
\begin{align*}
\pi^*_2 D\,\tilde{g}=\,\pi^*_2 D \left((1-\epsilon)\,\pi^*_2 g_0\right)=\, d_4(1-\epsilon)\pi^*_2g_0 +\,(1-\epsilon)\,\pi^*_2D \left(\pi^*_2g_0\right),
\end{align*}
where $d_4 (1-\epsilon)$ is the differential of the function $(1-\epsilon):J^2 M_4\to \mathbb{R}$ only respect to the elements of $TJ^4M_4$ complement of $\ker d\pi_2$.
The covariant derivative of the pull-back metric is
\begin{align*}
\pi^*_2D \left(\pi^*_2g_0\right)=\,\pi^*_2(D\,g_0)=\,0,
\end{align*}
by the metric condition $D\,g_0=0$ of the Levi-Civita connection.  Let us introduce the $1$-form $\mathcal{W} \in \Lambda ^1 \,J^2M_4$ by
 \begin{align*}
 \mathcal{W}: =\,\,\frac{1}{1-\epsilon}\,d_4 (1-\epsilon).
 \end{align*}
   Hence we have the relation between $\pi^*_2 D$ and $\tilde{g}$
\begin{align}
\pi^*_2 D \,\tilde{g}=\,\mathcal{W}\otimes \,\tilde{g}.
\label{non-metricity of the pull-back connection}
\end{align}
If $\{X_A,\,A=1,2,...,12\}$ is a local frame in $J^2M_4$, then we have
\begin{align*}
\pi^*_2 D_{X_A} \,\tilde{g}=\,\mathcal{W}(X_A)\,\,\tilde{g}.
\end{align*}

 In dimension $\dim (M_4)=4$, the number of connection coefficients
 \begin{align*}
 \{\widetilde{\Gamma}^i\,_{Aj},\,A=1,...,12;\,i,j=1,2,3,4\}
 \end{align*}
  of an arbitrary connection on $\pi^*_2 TM_4$ is $12*4*4=\,192$. However, there are constrains arising from the definition of the pull-back connection $\pi^*_2 D$. In particular, we have the conditions
 \begin{align}
 \pi^*_2 D_{X_a} \,\pi^* S =0,\,\quad \forall X_a\in\, \ker (d \pi_2),\, \,S\in \,\Gamma TM_4.
 \label{conditions on the kern of piD}
 \end{align}
Since $\dim (\ker ( d \pi_2))=8$, these constrains impose  $8*4*4=128$ pointwise  conditions on the connection coefficients of the pull-back connection $\pi^*_2 D$.

The non-metricity condition implies a number of $4 * 10=\,40$ independent constrains.

The number of independent constrains on the connection coefficients introduced by the torsion free relations $\widetilde{T}(X_i,X_j)=0$ is $4 *4*3*1/2=\,24$. Thus the total number of constrains is $128+24+40=\,192$, which is enough and sufficient to determine completely the connection $\pi^*_2 D$.
Therefore, the constrains \eqref{zero torsion for the pull-back condition}, \eqref{conditions on the kern of piD} and \eqref{non-metricity of the pull-back connection} are enough to determine the pull-back connection $\pi^*_2 D$.
\begin{proposicion}
The pull-back connection $\pi^*_2 D$ is the unique linear  connection $\widetilde{\nabla}$  on the bundle $p_1: \,\pi^*_2 TM_4 \to J^2M_4$ determined by the conditions
\begin{itemize}
\item   The constrains
\begin{align}
 \pi^*_2 D_{X} \,\pi^* S =0,\,\quad \forall \,S\in \,\Gamma TM_4
 \label{constrains from pull-back of D}
 \end{align}
for $\{X_a,\,a=1,...,8\}$ a local frame for each kernel $\ker ( d \pi_2|_u)$ with $u\in J^2M_4$.

\item Torsion-free condition of the form
\begin{align}
\widetilde{\nabla}_{X_i} \,\pi^*_2Z_j-\,\widetilde{\nabla}_{X_j} \,\pi^*_2Z_i-\pi^*_2([Z_i,Z_j])=0,
\end{align}
for $d(X_i)=Z_i$ local vector frames.
\item The following {\it pull-back} conditions
\begin{align}
\widetilde{\nabla}_{X} \,\tilde{g}=\,\mathcal{W}(X)\,\,\tilde{g},
\label{pull-back conditions}
\end{align}
for every $X \in \Gamma TJ^2 M_4$.
\end{itemize}
\label{axioms for the connection}
\end{proposicion}

We would like to remark the similarity of this construction with the construction of the Chern connection in Finsler geometry \cite{BaoChernShen}.
However, we did not require the introduction of a non-linear connection in $J^2M_4$ in order to characterize the connection by the relations discussed in proposition \ref{axioms for the connection}. Instead, we introduce the conditions \eqref{constrains from pull-back of D}, that together with torsion free and metric non-compatibility conditions, are enough to determine the connection consistently.

\section{Discussion}
 In this work, a formal theory for spacetimes whose metric structure admits a maximal acceleration has been investigated  in the framework of metrics depending on higher order jets. We have adopted the viewpoint that the existence of a maximal proper acceleration should be imprinted in the geometry of the spacetime, in a similar way as the existence of a maximal speed is imprinted in the causal structure of the spacetime. The consequence of this assumption is that the metric spacetime in our models is more malleable than in current models of classical physics and quantum field theory, since the metric of maximal acceleration depend upon the way the spacetime is tested.

The spacetimes of maximal acceleration that we have considered are not necessarily originated by corrections to general relativity due to a quantum origin of gravity or quantum mechanical frameworks. Indeed, we have discussed uniform upper bounds in the Riemannian curvature in the contest of {\it classical geometries of maximal acceleration}, where the fundamental notion of metric of maximal acceleration is a classical object, in contrast with arguments involving quantum gravity \cite{RovelliVidotto}, string theory \cite{ParentaniPotting, BowickGiddins, FrolovSnachez1991} or quantum mechanics \cite{Brandt1983, Brandt1989, Caianiello, CaianielloFeoliGasperiniScarpetta}.

 We have shown that the existence of an universal bound for the proper acceleration of test particles implies that certain combinations of the Riemannian curvature components must remain uniformly bound in certain sense in the whole region of the spacetime where the curvature is defined and finite.  The combinations of the components of the Riemannian curvature are bilinear combinations that involve certain Jacobi fields. Note that the result obtained does not apply to all the components of the curvature tensor, although if further symmetries are considered, it could led to bounds on combinations involving more curvature components. This result applies rather generally. This result applies as long as the metric acts on tangent vectors by the expression \eqref{maximalaccelerationmetric0} and the additional assumptions of section 3 hold. Specially relevant, is to be able to interpret the maximal acceleration as determining a minimal proper time in theories where local speed is uniformly bounded by the speed of light in vacuum, following the argument from Caldirola \cite{Caldirola1981}.

  It is worth good to remind here that the theory presented in section 2 differs from other theories where the proper acceleration is also uniformly bounded, as in Caianello's theory \cite{Caianiello,CaianielloFeoliGasperiniScarpetta} and Brandt's theory \cite{Brandt1983,Brandt1989}. In those theories, the upper bound on the proper acceleration has a quantum origin, while in our theory the starting point is a spacetime arena consistent with possible violations of Einstein clock hypothesis, and we end up with a theory for metrics of maximal acceleration as a pre-eminent case. Also, the values for the maximal acceleration are different in Caianiello's theory \cite{Caianiello,Caianiello1984}, Brandt's theory \cite{Brandt1983} and in classical electrodynamic theories with higher order fields, for instance as in \cite{Ricardo2017b}. On the other hand, it is direct that the formal expressions for the metrics coordinates essentially coincide with a covariant form of Caianiello's theory \cite{Ricardo2007}. Therefore, if the requirements of section 3 are meet in Caianello's and Brandt's theories, then our theory developed in section 3 can apply to them, and the combinations of curvature discussed in section 3 will apply. Specifically, two requirements must hold.  The first is that in the formal limit $A^2_{\mathrm{max}}\to +\infty$, it holds that $g\to g_0$, where $g_0$ is a spacetime metric structure compatible with clock hypothesis. In name of simplicity, this structure has been assumed to be a Lorentzian metric. This requirement, is meet by Caianiello's and Brandt's theories of maximal acceleration. Second, the theories where the results of section 3 apply, must be compatible with the hypothesis of an uniform minimal lapse of proper time. It is known that in Brandt's theory and in Caianiello's theory, being relativistic theories in the above sense of being consistent with the existence of a maximal speed, they are also consistent with the existence of an uniform minimal lapse of proper time.  In the case of Caianiello's theory, this was discussed in \cite{Caianiello1984}, while for the case of Brandt's theory make use of the relation of Compton wave length and Schwarzschild radius \cite{Brandt1983}, which is implicitly consistent with a minimal lapse of time.

In section 5, a geometric framework for spacetimes endowed with a metric of maximal acceleration has been developed. The driving principle is a new version of the equivalence principle, that, although less strong than Einstein equivalence principle, is naturally fitted for spacetimes with a higher order jet metric, including Finsler spacetimes \cite{GallegoPiccioneVitorio:2012,Ricardo2017} and the spacetimes considered in this paper. Then we have introduced a connection and show how a geometric theory based on such connection full-fills the new form of the equivalence principle and how the connection is determined by natural conditions and the form of the metric of maximal acceleration.

The theory of section 5 has been developed thinking in the theory of spacetimes endowed with a maximal proper acceleration discussed in section 2, but one can raise the question if the geometric theory of section 5 can also be applied to  Caianiello's theory \cite{Caianiello,CaianielloFeoliGasperiniScarpetta} and to Brandt's theory \cite{Brandt1983,Brandt1989}.As we mentioned before, since the metric structures are formally the same in these theories (if we consider the covariant form of Caianiello's theory discussed in \cite{Ricardo2007}, one can apply the same methods to these theories and the geometric theory developed in section 5 is applicable. However, one needs to say that such interpretation of Caianiello's and Brandt's theories is not as compelling, because section 5 is mainly driven by our reformulation of the equivalence principle discussed in section 4, which was driven by the formal structure of our theory of maximal acceleration discussed in section 2.

A criticism that is commonly raised to theories with a maximal proper acceleration is that, although the idea is very fertile in consequences, there is no empirical evidence for it. However, a new approach to observe effects of maximal acceleration in laser-plasma dynamics has been recently discussed \cite{Ricardo2019} in the contest of the higher order jet electrodynamics developed by the author \cite{Ricardo2012,Ricardo2015,Ricardo2017b}. In ref. \cite{Ricardo2019} it was argued that several modifications of the Lorentz force due to maximal acceleration are potentially testable in near future laser-plasma acceleration facilities. A similar idea, but developed in the contest of Caianiello's models was developed in the reference \cite{FeoliLambiasePapiniScarpetta1997}, leading to different testable predictions than in \cite{Ricardo2019}. Confirmation of such qualitative dynamical effects in modern laser-plasma facilities could show for first time deviations from a local Lorentzian structure of spacetime.

Other developments of our theory are the following. Among the conditions that characterize the connection $\pi^*_2 D$ discussed in section 5  is the non-metricity property \eqref{non-metricity of the pull-back connection}. Such non-metricity property implies consequences at a cosmological and astrophysical level, leading to constraints on the value of the maximal acceleration, in a similar way, for instance, as the non-metricity due to spacetime defects leads to constraints on their density distribution of spacetime defects in FRW-like models \cite{Hossenfelder-Ricardo2018}. In the case of models where the maximal acceleration is fixed by the theory, such constraints will become genuine predictions. However, in our that this type of research can be developed in full deepness, a sound theory of gravitation in the framework of spacetimes with higher order jet geometry, the corresponding field equations and the FRW-like solutions is need to be developed first.

As we have mentioned before, a fundamental point missing in this paper and generally, in our approach to higher order fields theories, is the formulation a of complete gravitational model in the framework of spacetimes with a  maximal acceleration. One can partially address this issue by considering models of maximal acceleration already developed by E. Caianiello and co-workers \cite{CaianielloFeoliGasperiniScarpetta}, that although based upon a mathematical formalism which is not general covariant and based upon different assumptions than higher order jet field theory as discussed in \cite{Ricardo2012} and in this paper, it could be an inspiring source towards a consistent theory of spacetimes with metrics of maximal acceleration. However, we think that as happened with the development of  general relativity, only the combination of a powerful set of physical principles expressed through  mathematical rigourous notions and methods could bring light on the new field equations of motion for spacetimes with metrics of maximal acceleration.

Other problems to be considered in further research include a comprehensive treatment of singularities and its eventual resolution in the framework of metrics with maximal acceleration.
Also, the aspects of the thermodynamics of black holes that should be modified by the existence of a maximal acceleration have not been considered in this paper.

\subsection*{Acknowledgements} I would like to acknowledge to S. Vergara Molina for the continuous interest in the theory of maximal acceleration.


\begin{thebibliography}{22}

\bibitem{BaoChernShen} D. Bao, S.S. Chern and Z. Shen, {\it An Introduction to Riemann-Finsler
Geometry}, Graduate Texts in Mathematics 200, Springer-Verlag (2000).

\bibitem{Beem1970} J. K. Beem, {\it Indefinite Finsler Spaces and Timelike Spaces},
Canad. J. Math. {\bf 22}, 1035 (1970).

\bibitem{BowickGiddins} M. J. Bowick and S. B. Giddins, {\it High-Temperature Strings}, Nucl. Phys. B {\bf 325}, 631 (1989).

\bibitem{BozzaFeoliPapiniScarpetta} V. Bozza, A. Feoli, G. Papini and G. Scarpetta, {\it Maximal acceleration effects in
    Reissner-Nordstr$\ddot{o}m$ space}, Phys. Lett. A {\bf 271}, 35 (2000).

\bibitem{BozzaFeoliLambiasePapiniScarpetta} V. Bozza, A. Feoli, G. Lambiase, G. Papini and G. Scarpetta, {\it Maximal
    acceleration effects in Kerr space}, Phys. Lett. A {\bf 283}, 53 (2001).

 \bibitem{Brandt1983} H. E. Brandt, {\it Maximal proper acceleration relative to the vacuum}, Lett. Nuovo. Cimento {\bf 38}, 522 (1983).

 \bibitem{Brandt1989} H. E. Brandt, {\it Maximal proper acceleration and the structure of spacetime}, Found. of Phys. Lett., vol. {\bf 2}, 39 (1989).

 \bibitem{Caianielloquantum} E. R. Caianiello, {\it Geometry from quantum theory}, Il Nuovo Cimento, Vol. {\bf 59} B, 350-466 (1980);
    E.R. Caianiello, {\it Quantum and Other Physics as Systems Theory}, La Rivista
del Nuovo Cimento, Vol. {\bf 15}, Nr 4 (1992).


 \bibitem{Caianiello} E. R. Caianiello, {\it Is there a Maximal Acceleration}, Lett. Nuovo Cimento {\bf 32}, 65 (1981).

 \bibitem{Caianiello1984} E. R. Caianiello, {\it Maximal Acceleration as a Consequence of Heisenberg's Uncertainty Relations}, Lett. Nuovo Cimento {\bf 41}, 370 (1984).

 \bibitem{CaianielloFeoliGasperiniScarpetta} E. Caianiello, A. Feoli, M. Gasperini and G. Scarpetta, {\it Quantum
    corrections to the spacetime metric from geometric phase space quantization}, Int. J. Theor. Phys. {\bf 29}, 131 (1990).

\bibitem{CaianielloGasperiniScarpetta} E. R. Caianiello, M. Gasperini and G. Scarpetta, {\it Inflaction and singularity
    prevention in a model for extended-object-dominated cosmology}, Clas. Quan. Grav. {\bf 8}, 659 (1991).

\bibitem{Caldirola1956} P. Caldirola, {\it A new model of classical electron}, Suplemento al Nuovo Cimento, Vol. III,
    Serie X, 297-343 (1956).

\bibitem{Caldirola1981} P. Caldirola, {\it On the existence of a maximal acceleration in the relativistic theory of the
    electron}, Lett. Nuovo. Cim. vol. {\bf 32}, N. 9, 264 (1981).


\bibitem{Einstein1922} A. Einstein, {\it The meaning of relativity}, Princenton University Press (1923).

\bibitem{FeoliLambiasePapiniScarpetta} A. Feoli, G. Lambiase, G. Papini and G. Scarpetta, {\it Classical Electrodynamics of a particle with Maximal Acceleration Corrections}, N.
Cim. {\bf 112}B, 913 (1997).

\bibitem{FeoliLambiasePapiniScarpetta1997} A. Feoli, G. Lambiase, G. Papini and G. Scarpetta, {\it Schwasrzschild field with
    maximal acceleration corrections}, Phys. Lett. A  {\bf 263}, 147 (1999).
    
\bibitem{FrolovSnachez1991} V.P. Frolov and N. Sanchez, {\it Instability of Accelerated Strings and the Problem of Limiting
Acceleration}, Nucl. Phys. B {\bf 349}, 815, (1991).

\bibitem{Ricardo2007} R. Gallego Torrom\'e, {\it On a covariant version of Caianiello's Model}, Gen. Rel. Grav. {\bf 39}, 1833-1845 (2007).

\bibitem{GallegoPiccioneVitorio:2012} R. Gallego Torrom\'e, P. Piccione and H. Vit\'orio,
{\it On Fermat's principle for causal curves in time oriented {Finsler} spacetimes},
{J. Math. Phys.} {\bf 53}, 123511 (2012).

\bibitem{Ricardo2012} R. Gallego Torrom\'e, {\it Geometry of generalized higher
order fields and applications to classical linear electrodynamics}, arXiv:1207.3791.

\bibitem{Ricardo2015} R. Gallego Torrom\'e, {\it An effective theory of metrics with maximal proper acceleration}, Class. Quantum. Grav. {\bf 32}, 2450007 (2015).

    \bibitem{Ricardo2017b} R. Gallego Torrom\'e, {\it A second order differential equation for a point
Charged particle}, International Journal of Geometric Methods in Modern Physics,
Vol. {\bf 14}, No. 04, 1750049 (2017).

\bibitem{Ricardo2017} R. Gallego Torrom\'e, {\it On singular generalized Berwald spacetimes and the
equivalence principle}, International Journal of Geometric Methods in Modern
Physics Vol. {\bf 14}, No. 06, 1750091 (2017).


\bibitem{Ricardo Nicolini 2018} R. Gallego Torrom\'e and P. Nicolini, {\it Theories with maximal acceleration}, International Journal of Modern Physics A Vol. {\bf 33}, 1830019 (2018).


\bibitem{Ricardo2019} R. Gallego Torrom\'e, {\it Some consequences of theories with maximal acceleration in laser-plasma acceleration}, Modern Physics Letters A Vol. {\bf 34} (2019) 1950118.

\bibitem{Gasperini 1987} M. Gasperini, {\it Very early cosmology in the maximal acceleration
hypothesis}, Astrophys. Space Sci. {\bf 138}, 387 (1987)

\bibitem{Gasperini 1991}M. Gasperini, {\it A geometrical regularization procedure for the curvature of
cosmological background},  in Proc. of the Workshop on  ''Advances in
theoretical
physics" (Vietri, October 1990), ed.  E. R. Caianiello (World
Scientific, Singapore, p. 77 (1991).

\bibitem{Hartman} P. Hartman, {\it Ordinary differential equations}, Birkh\"auser (1982).


\bibitem{Hawking Ellis 1973} S. W. Hawking and G. F. R. Ellis, {\it The Large Scale Structure of the Universe}, Cambridge University Press (1973).

\bibitem{Hossenfelder-Ricardo2018} S. Hossenfelder and R. Gallego Torrom\'e, {\it General relativity with space-time defects}, Class. Quantum Grav. Vol. {\bf 35}, 175014 (2018).

\bibitem{Mashhoon1990}  B. Mashhoon,
{\it Limitations of spacetime measurements}, Phys. Lett. A {\bf 143}, 176-182 (1990); B. Mashhoon,
{\it The Hypothesis of Locality in relativistic physics}, Phys. Lett. A {\bf 145}, 147 (1990).

\bibitem{ParentaniPotting} R. Parentani and R. Potting, {\it Accelerating Observer and the Hagedorn 
Temperature}, Phys. Rev. Lett.{\bf 63}, 945 (1989).


\bibitem{RovelliVidotto} C. Rovelli and F. Vidotto, {\it Maximal acceleration in covariant loop gravity and singularity
    resolution}, Phys. Rev. Lett. {\bf 111}, 091303 (2013).


\bibitem{ThorneLeeLightman} K. S. Thorne, D. L. Lee and A. P. Lightman, {\it Foundations for a theory of gravitational theories}, Phys. Rev. D. {\bf 7}, 3563 (1973).

\bibitem{Syngespecial1965} {J. L. Synge}, {\it Relativity: The {Special} {Theory}}, {North-Holland Publishing Company, Amsterdam} (1965).


\end{thebibliography}
\end{document}